\documentclass[11pt]{article}
\usepackage{authblk,algorithm,amsmath,amssymb,amsthm,caption,color,epsfig,latexsym,listings,geometry,graphics,graphicx,psfrag,url}
\usepackage[utf8]{inputenc}
\usepackage[noend]{algpseudocode}
\graphicspath{ {Figures/} }

\newtheorem{theorem}{Theorem}[section]
\newtheorem{corollary}{Corollary}[section]
\newtheorem{proposition}{Proposition}[section]
\newtheorem{definition}{Definition}[section]
\newtheorem{lemma}{Lemma}[section]
\newtheorem{rmk}{Remark}[section]

\usepackage{xpatch}
\xapptobibmacro{finentry}{\par\printfield{NOTE}}{}{}

\usepackage{multicol}
\usepackage{stmaryrd} %for \llbracket and \rrbracket
\newcommand{\Jac}{{\text{Jac}}}
\newcommand{\EE}{\mathbb{E}}

\newcommand{\PP}{{\mathbb{P}}}
\newcommand{\RR}{\mathbb{R}}

\newcommand{\curlyX}{{X}}

\title{Metric Dimension and Resolvability of Jaccard Spaces}

\author[1,*]{Manuel E. Lladser}
\author[1]{Alexander J. Paradise}

\affil[1]{Department of Applied Mathematics, University of Colorado, Boulder, USA}
\affil[*]{Corresponding author: {manuel.lladser@colorado.edu}}

\date{}

\begin{document}

\maketitle

\abstract{
A subset of points in a metric space is said to resolve it if each point in the space is uniquely characterized by its distance to each point in the subset. In particular, resolving sets can be used to represent points in abstract metric spaces as Euclidean vectors. Importantly, due to the triangle inequality, points close by in the space are represented as vectors with similar coordinates, which may find applications in classification problems of symbolic objects under suitably chosen metrics. In this manuscript, we address the resolvability of Jaccard spaces, i.e., metric spaces of the form $(2^X,\text{Jac})$, where $2^X$ is the power set of a finite set $X$, and $\text{Jac}$ is the Jaccard distance between subsets of $X$. Specifically, for different $a,b\in 2^X$, $\text{Jac}(a,b)=|a\Delta b|/|a\cup b|$, where $|\cdot|$ denotes size (i.e., cardinality) and $\Delta$ denotes the symmetric difference of sets. We combine probabilistic and linear algebra arguments to construct highly likely but nearly optimal (i.e., of minimal size) resolving sets of $(2^X,\text{Jac})$. In particular, we show that the metric dimension of $(2^X,\text{Jac})$, i.e., the minimum size of a resolving set of this space, is $\Theta(|X|/\ln|X|)$. In addition, we show that a much smaller subset of $2^X$ suffices to resolve, with high probability, all different pairs of subsets of $X$ of cardinality at most $\sqrt{|X|}/\ln|X|$, up to a factor.

\medskip

}

%%%%%%%%%%%%%%%%%%%%%%%%%%%%%%%%%%%%%%%
\section{Introduction}
%%%%%%%%%%%%%%%%%%%%%%%%%%%%%%%%%%%%%%%

A metric space is an ordered-pair of the form $(X,d)$, where $X$ is a nonempty set, and $d:X\times X\to\mathbb{R}$ a function satisfying that $d(x,y)=d(y,x)\ge0$, $d(x,y)=0$ if and only if $x=y$, and $d(x,y)\le d(x,z)+d(z,y)$, for all $x,y,z\in X$. In particular, $d$ is non-negative, symmetric, and satisfies the triangular inequality. We say the metric space is finite when $|X|<+\infty$.

Resolvability extends the concept of trilateration of the plane to general metric spaces; in particular, it includes the vertex set of connected graphs endowed with shortest path distances between vertices---which is where the concept originated~\cite{ErdHarTut65,Sla75,HarMel76}. In a metric space $(X,d)$, a non-empty set $R=\{r_i:i\in I\}\subset X$, with $I=\big\{1,\ldots,|R|\big\}$, is said to \textit{resolve} it when the transformation 
\begin{equation}
d(x|R):=\big(d(x,r_i)\big)_{i\in I}, \text{ for each }x\in X,
\label{def:dxR}
\end{equation}
is one-to-one. In particular, $d(\cdot|R)$ uniquely encodes points in $X$ as $|R|$-dimensional real vectors; and, owing to the triangular inequality, proximate points in $X$ are encoded as vectors with similar coordinates. Resolving sets thus enable sound embeddings of metric spaces into Euclidean ones, which can be useful for generating numerical features of symbolic objects in statistical and machine learning tasks like regression or classification~\cite{TilLla18,RutLla21}.

One can think of a resolving set as a collection of ``landmarks'' in a metric space that uniquely identify the ``location'' of any point in that space by its distance to those landmarks. In that regard, resolvability serves as a form of ``multi-lateration'' of the space, similar to tri-lateration, although more than three landmarks may be needed to resolve a given metric space. 

Irrespective of the metric space, resolving sets always exist, although they are never unique in non-trivial settings. This is because $X$ always resolves $(X,d)$, and if $R$ resolves $(X,d)$ and $S\supset R$, then $S$ also resolves it. So, finding a resolving set is straightforward. In contrast, finding a resolving set with the smallest possible size is usually challenging; in fact, it is an NP-complete problem in arbitrary finite metric spaces~\cite{Kar72,GarJoh79}. Minimizing the size of a resolving set is nonetheless crucial to embedding the points in $X$ into a low-dimensional Euclidean space using transformations of the form~(\ref{def:dxR}). This motivates the notion of \textit{metric dimension}, which is the size of the smallest resolving set of a metric space $(X,d)$, denoted from now on as $\beta(X,d)$.

For a concise overview of resolvability and metric dimension in the context of graph theory, see~\cite{TillFroLla19}. Instead, for a comprehensive review of these and related concepts, see~\cite{TilFroLla23,KuzYer21}.

A very limited number of studies have addressed the resolvability of non-graphical metric spaces in the literature~\cite{Mur75,BauBea13}, as most efforts have focused on finite graphs~\cite{ChaEroJohEtAl00}. Nevertheless, spaces with metric dimensions 1 or 2 have been characterized under general topological assumptions~\cite{Mur75, BauBea13}. It is also known that the metric dimension of a $k$-dimensional subspace of $\mathbb{R}^n$ with respect to the Euclidean distance is $(k+1)$; in particular, $(\mathbb{R}^n,\|\cdot\|_2)$ has metric dimension $(n+1)$~\cite{BauBea13}. The hypersphere $(\mathbb{S}^n,\|\cdot\|_2)$ has also metric dimension $(n+1)$. Additionally, the metric dimension of the hyperbolic space $\mathbb{H}^n$ with respect to the metric $d(x,y):=\int_x^y dx/x_n$, for all $x,y\in\mathbb{H}^n$, is $(n+1)$~\cite{BauBea13}. Likewise, the metric dimension of the $n$-dimensional unit ball $\mathbb{B}^n$ with respect to the metric $d(x,y):=\int_x^y 2\,|dx|/(1-\|x\|^2)$, with $x,y\in\mathbb{B}^n$, is $(n+1)$~\cite{BauBea13}.

In contrast, the systematic study of the resolvability and metric dimension of non-graphical, finite, metric spaces is essentially unexplored. In this paper, we study the resolvability of finite Jaccard metric spaces,  i.e., metric spaces of the form $(2^X,\Jac)$, where $2^X$ denotes the power set of a finite set $X$, and $\Jac$ is the \textit{Jaccard distance} between subsets of $X$~\cite{Jac01}. Namely, for all $a,b\in2^X$,
\[\Jac(a,b):=
\begin{cases}
\frac{|a\Delta b|}{|a\cup b|}, & a\ne b;\\
0, & a=b.
\end{cases}
\]
$\Jac$ is a metric in $2^X$~\cite{Gil72, Kos19}. (In the literature, for distinct $a,b\in 2^X$, the quantity $1-\Jac(a,b)=|a\cap b|/|a\cup b|$ is referred to as the \textit{Jaccard similarity}. This index is widely used in fields such as information retrieval%~\cite{Cha16}
, data mining, and natural language processing, among many others.) 

Given that $X$ is finite in our setting, we may, in principle, estimate $\beta(2^X,\Jac)$ and find non-trivial resolving sets with the so-called Information Content Heuristic (ICH)~\cite{HauSchVie12}. In a general setting, the input of this algorithm is the (symmetric) distance matrix between all pairs of points in a metric space, and the output is a subset of columns that resolve it, which is determined greedily through an entropy maximization procedure. Unfortunately, however, in the context of Jaccard spaces, the ICH is infeasible even for moderate values of $|X|$ because of its $O(2^{3|X|})$ time complexity.

Nevertheless, besides being of theoretical interest, learning to resolve optimally or nearly optimally Jaccard spaces may find applications in e.g. lexicon-based approaches to natural language processing (NLP). In the most basic implementation of this idea, $X$ would be the set of all words in a language and sentences represented as subsets of $X$ (aka, \textit{bag of words}). The Jaccard distance is then a natural way to assess the similarity of sentences based on the words used, and a resolving set would induce a numerical encoding of sentences, mapping sentences with similar word content into vectors with similar coordinates, potentially providing low-dimensional feature vectors to learn to classify or regress sentences based on their lexicon~\cite{Par24}.

%%%%%%%%%%%%%%%%%%%%%%%%%%%%%%%%%%%%%%%
\subsection{Main Results} 
%%%%%%%%%%%%%%%%%%%%%%%%%%%%%%%%%%%%%%%

\textbf{In what remains of this manuscript, $X$ is assumed to be a finite non-empty set.}

In this section, we outline our key findings, with expanded statements and proofs provided in Section~\ref{sec:Res&Pro}.

From now on, the Jaccard distance is the reference metric in $2^X$; in particular, e.g., statements like \textit{``$R$ resolves $2^X$,''} mean that \textit{``$R\subset 2^X$ resolves $(2^X,\Jac)$.''} We also say that $R$ \text{resolves} $a,b\in 2^X$ when there exists $r\in R$ such that $\Jac(a,r)\ne\Jac(b,r)$.

We first provide a necessary condition for a set $R$ to resolve $2^X$. 

\begin{proposition}
If $R$ resolves $2^X$ then $R$ separates the distinct elements of $\curlyX$, and it covers all but possibly one element in $X$.
\label{prop:main}
\end{proposition}

The proof of the proposition can be found in Section~\ref{subsec:necessary}. We note that these properties are necessary but not sufficient. For instance, if $\curlyX=\{1,2,3,4\}$ and $R=\big\{\{1,2\},\{1,3\},\{1,4\}\big\}$ then $R$ separates different elements in $\curlyX$ and also covers it. Nevertheless, $R$ is not resolving because $\Jac(\{1\}|R)=(1/2,1/2,1/2)=\Jac(\curlyX|R)$. This counterexample can be easily generalized to sets $X$ of arbitrary size.

Next, we provide a lower bound on the size of any resolving subset of $2^X$; in particular, this is also a lower bound for $\beta(2^X,\Jac)$.

\begin{proposition}
If $R$ resolves $2^X$ then 
\[|R|\ge\frac{|X|(\ln 2)\left(1+o(1)\right)-2\ln\left(|X|/2\right)-1}{\ln\left(|X|/2+1\right)}\sim\frac{|X|\ln2}{\ln|X|}.\]
\label{prop:lb1}
\end{proposition}

The proof of the proposition can be found in Section~\ref{subsec:lb}. 

To state our main two results we require the following definition.

\begin{definition}
A random $r\in 2^X$ is said to have a Binomial$(X,1/2)$ distribution, in which case we write $r\sim\text{Binomial}(X,1/2)$, when $\PP(x\in r)=1/2$ for each $x\in X$, and the events $[x\in r]$ with $x\in X$ are independent.
\end{definition}

Clearly, if $r\sim\text{Binomial}(X,1/2)$ then $|r|\sim\text{Binomial}(|X|,1/2)$; namely, $\PP\big(|r|=k\big)=\frac{1}{2^{|X|}}{|X|\choose k}$ for $k=0,\ldots,|X|$.

\begin{theorem}
\label{thm:main1}
If $k\ge\frac{2\ln(2e)|X|}{\ln(|X|/2)}$ and $r_1,\ldots,r_k\sim\text{Binomial}(X,1/2)$ are independent and identically distributed (i.i.d.), then, for each $x\in X$, $R:=\big\{\emptyset,\{x\},X\setminus\{x\},r_1,\ldots,r_k\big\}$ resolves $2^X$, with overwhelmingly high probability, as $|X|\to\infty$. 
\end{theorem}

The proof of the theorem can be found in Section~\ref{subsec:equalcard} and relies on auxiliary results in Sections~\ref{subsec:bummeryay} and~\ref{subsec:Equidistant}.

In conjunction, Proposition~\ref{prop:lb1} and Theorem~\ref{thm:main1} imply that
\[
\frac{(\ln 2) |X|}{\ln(|X|/2)} \left(1 + o(1)\right) \le \beta(2^X, \Jac) \le \frac{2 \ln(2e) |X|}{\ln(|X|/2)} \left(1 + o(1)\right),
\]
which characterizes the metric dimension of $2^X$ with respect to the Jaccard distance within a factor of $\frac{2 \ln(2e)}{\ln 2} \approx 5.0$. In particular, we can assert the following.

\begin{corollary}
$\beta(2^X,\Jac)=\Theta\left(\frac{|X|}{\ln|X|}\right)$, as $|X|\to\infty$.
\end{corollary}

It turns out that, for any $x\in X$, the set $\big\{\emptyset,\{x\},X\setminus\{x\}\big\}$ resolves all pairs of subsets of $X$ with different cardinalities (see Lemma~\ref{lem:bummeryay} ahead). So the crux of the proof of Theorem~\ref{thm:main1} lies in showing that the sets in $\{r_1, \ldots, r_k\}$ resolve all possible pairs $a,b\in 2^X$ of equal size---with overwhelmingly high probability---when $|X|$ is large. We demonstrate this in Section~\ref{subsec:equalcard}.

In the context of potential NLP applications outlined in the Introduction, it is unclear whether the highly likely resolving set proposed in Theorem~\ref{thm:main1} is of any practical value for distinguishing between bags-of-words of different cardinalities. This is because the numerical encoding in (\ref{def:dxR}) based on this set might differentiate such pairs solely based on the presence or absence of a single word or token, which seems too coarse for practical use in NLP classification (or regression) problems. Our following result addresses this issue by proposing a less contrived set, which is likely to resolve all pairs of bags-of-words of different cardinalities. Its proof can be found in Section~\ref{subsec:diffcard}.

\begin{theorem}
\label{thm:main2}
Let $\epsilon>0$. If $k\ge(4+\epsilon)\sqrt{|X|}$ and 
$r_1,\ldots,r_k\sim\text{Binomial}(X,1/2)$ are i.i.d., then $R:=\big\{r_1,r_1^c,\ldots,r_k,r_k^c\big\}$ resolves all pairs of subsets of $X$ of different size, with overwhelmingly high probability, as $|X|\to\infty$.
\end{theorem}

As expected, the lower bound for the size of the set $R$ in Theorem~\ref{thm:main2} is asymptotically negligible compared to the one in Theorem~\ref{thm:main1}; after all, the former set is only required to resolve pairs of subsets of $X$ with different cardinalities, which, as explained earlier, can be accomplished using just three subsets of $X$ (i.e., the empty set, and any singleton and its complement). Nevertheless, in practical situations---for instance, when representing social media posts as bags-of-words---more often than not, a random pair of posts would be associated with bags-of-words of different cardinality. In particular, in terms of the numerical encoding in (\ref{def:dxR}), Theorem~\ref{thm:main2} suggests that $O(\sqrt{|X|})$ Jaccard distances, as opposed to $\Theta\left(|X|/\ln|X|\right)$, should suffice in practice to encode posts effectively when the reference lexicon $X$ is sufficiently large. Our following result makes this intuition precise at the expense of limiting the size of bags-of-words one wishes to resolve.

\begin{corollary}
\label{cor:main3}
Let $0<\epsilon<1$. If $k\ge(4+\epsilon)\sqrt{|X|}$ and 
$r_1,\ldots,r_k\sim\text{Binomial}(X,1/2)$ are i.i.d., then the set $R:=\big\{r_1,r_1^c,\ldots,r_k,r_k^c\big\}$ resolves all different pairs of subsets of $X$ of size at most $\frac{(1-\epsilon)(\ln\pi)\sqrt{|X|}}{\ln|X|}$, with overwhelmingly high probability, as $|X|\to\infty$.
\end{corollary}

%%%%%%%%%%%%%%%%%%%%%%%%%%%%%%%%%%%%%%%
\section{Technical Results and Proofs}
\label{sec:Res&Pro}
%%%%%%%%%%%%%%%%%%%%%%%%%%%%%%%%%%%%%%%

%%%%%%%%%%%%%%%%%%%%%%%%%%%%%%%%%%%%%%%
\subsection{Necessary Conditions for Resolvability}
\label{subsec:necessary}
%%%%%%%%%%%%%%%%%%%%%%%%%%%%%%%%%%%%%%%

In this section, we prove Proposition~\ref{prop:main}. Specifically, suppose that $R$ resolves $2^\curlyX$. Next we show that the following properties applies:
\begin{enumerate}
\item[(i)] For all $x_1, x_2 \in \curlyX$ with $x_1 \ne x_2$, there exists $r \in R$ such that either $x_1 \in r$ and $x_2 \notin r$, or $x_1 \notin r$ and $x_2 \in r$.
\item[(ii)] If $\emptyset \notin R$, then $R$ covers $\curlyX$, i.e., $\bigcup_{r \in R} r = \curlyX$.
\item[(iii)] If $\emptyset \in R$, then $R$ covers $\curlyX$, or there exists $x \in \curlyX$ such that $\bigcup_{r \in R} r = \curlyX \setminus \{x\}$.
\end{enumerate}

To show the property (i), suppose by contradiction that there are distinct $x_1,x_2\in\curlyX$ such that, for each $r\in R$, $\{x_1,x_2\}\subset r$ or $\{x_1,x_2\}\subset \curlyX\setminus r$. In the first case: $\Jac(\{x_1\},r)=1-1/|r|=\Jac(\{x_2\},r)$, and in the second case: $\Jac(\{x_1\},r)=1=\Jac(\{x_2\},r)$. In either case, $R$ could not possibly be resolving, which shows the first property.

To show the property (ii), suppose that there is $x\in\curlyX$, which does not belong to any of the sets in $R$. Then, for each $r\in R$, $\Jac(\{x\},r)=1=\Jac(\emptyset,r)$, which is not possible. This shows the second property. 

Finally, to show the property (iii), suppose there are distinct $x_1,x_2\in\curlyX$ which do not belong to any of the sets in $R$. Then, for each $r\in R$, $\Jac(\{x_1\},r)=1=\Jac(\{x_2\},r)$, which is not possible and completes the proof of the proposition.

%%%%%%%%%%%%%%%%%%%%%%%%%%%%%%%%%%%%%%%
\subsection{Metric Dimension lower bound}
\label{subsec:lb}
%%%%%%%%%%%%%%%%%%%%%%%%%%%%%%%%%%%%%%%

In this section, we prove Proposition~\ref{prop:lb1}. 

Suppose that $R$ resolves $2^\curlyX$. If $c,r\subset X$, then by the Inclusion-Exclusion Principle, $\Jac(c,r)=1-\frac{|c\cap r|}{|c|+|r|-|c\cap r|}$. Since $0 \le |c \cap r|\le|c|$, the range of $\Jac(\cdot|R)$, when restricted to sets $c$ such that $|c| = n$, has size at most $(n+1)^{|R|}$. In particular, due to the Pigeonhole Principle, we must have $\binom{|\curlyX|}{n}\le(n+1)^{|R|}$, i.e.:
\begin{equation}
|R|
\ge\max\limits_{0<n<|\curlyX|}\frac{\ln{|\curlyX|\choose n}}{\ln(n+1)}
\ge\frac{\ln{|\curlyX|\choose\lfloor|X|/2\rfloor}}{\ln\big(|X|/2+1\big)}.
\label{ine:explLB}
\end{equation}

The right-most lower bound above should be a reasonable estimate of the best one (based on the Pigeon Principle) because $\ln(n+1)$ is a slowly increasing function of $n$, and ${|\curlyX|\choose n}$, with $0<n<|X|$, is maximized at $n=\lfloor|X|/2\rfloor$ (equivalently, $n=\lceil|X|/2\rceil$). To make the last numerator above more explicit, we use that~\cite[Exercise 24, \S 1.2.5]{Knu97}:
\[\frac{n^n}{e^{n-1}}\le n!\le\frac{n^{n+1}}{e^{n-1}},\text{ for }n\ge1.\]
In particular, if $n=\lfloor|X|/2\rfloor$ then 
\begin{align*}
\ln\binom{|\curlyX|}{\lfloor|X|/2\rfloor}
&\ge |X|\ln|X|-(n+1)\ln(n)-\big(|X|-n+1\big)\ln\big(|X|-n\big)-1\\
&= |X|\left\{\ln 2+\frac{n}{|X|}\ln\left(\frac{|X|}{2n}\right)+\frac{|X|-n}{|X|}\ln\left(\frac{|X|}{2|X|-2n}\right)\right\}-\ln\Big(n\big(|X|-n\big)\Big)-1\\
&\ge |X|\big\{\ln 2+o(1)\big\}-2\ln\left(\frac{|X|}{2}\right)-1.
\end{align*}
The proposition is now direct from (\ref{ine:explLB}).

%%%%%%%%%%%%%%%%%%%%%%%%%%%%%%%%%%%%%%%
\subsection{Resolving Subsets of $X$ of Different Cardinalities}
\label{subsec:bummeryay}
%%%%%%%%%%%%%%%%%%%%%%%%%%%%%%%%%%%%%%%

\begin{lemma}
\label{lem:bummeryay}
For all $x\in X$ and all $a,b\in 2^X$, if $|a|\ne|b|$ then $a$ and $b$ are resolved by $R=\big\{\emptyset,\{x\},X\setminus\{x\}\big\}$.
\end{lemma}

\begin{proof}
Without any loss of generality assume that $|X|>1$. Fix an $x\in X$ and note  that for each $c\in 2^X$:
\begin{align*}
\Jac(c,\{x\})
&=1-\begin{cases}
\frac{1}{|c|}, & x\in c;\\
0, & x\notin c;
\end{cases}\\
\Jac(c,X\setminus\{x\})
&=1-\begin{cases}
\frac{|c|-1}{|X|}, & x\in c;\\
\frac{|c|}{|X|-1}, & x\notin c.
\end{cases}
\end{align*}

Define $R:=\big\{\emptyset,\{x\},X\setminus\{x\}\big\}$. Consider $a,b\in 2^X$ such that $|a|\ne|b|$, and suppose that $\Jac(a,r)=\Jac(b,r)$, for all $r\in R$. In particular, $a$ cannot be empty; otherwise, $
\Jac(b,\emptyset)=\Jac(a,\emptyset)=0$, implying that $b=\emptyset$ because $\Jac$ is a metric. However,
the latter is not possible because $|a|\ne|b|$. Likewise, $b$ is cannot be empty.

Moreover, if $x \in a$, then $\Jac(b, {x}) = 1 - |a|^{-1} < 1$. In particular, $x$ must be in $b$ as otherwise $\Jac(b, {x}) = 1$, which is not possible. But then, $1-|b|^{-1}=1-|a|^{-1}$, i.e., $|b|=|a|$, which is not possible either. Instead, if $x\notin a$ and $x\notin b$ then, because $\Jac(a,X\setminus\{x\})=\Jac(b,X\setminus\{x\})$, we must have that $1-\frac{|a|}{|X|-1}=1-\frac{|b|}{|X|-1}$, i.e., $|a|=|b|$, which is again not possible. Hence, there has to be an $r\in R$ such that $\Jac(a,r)\ne\Jac(b,r)$, implying that $R$ resolves $a$ and $b$. The same conclusion applies if $x\in b$, which completes the proof of the lemma.
\end{proof}

%%%%%%%%%%%%%%%%%%%%%%%%%%%%%%%%%%%%%%%
\subsection{Inner product Characterization of Equidistant Sets}
\label{subsec:Equidistant}
%%%%%%%%%%%%%%%%%%%%%%%%%%%%%%%%%%%%%%%

Two sets $a,b\in 2^X$ are said \textit{equidistant} from an $r\in 2^X$ when $\Jac(a,r)=\Jac(b,r)$. In this case, $r$ is not useful to resolve $a$ from $b$ when $a\ne b$, and we say that $a$ and $b$ \textit{collide} in terms of their Jaccard distance to $r$. 

In this section, we characterize collisions in linear algebra terms by representing subsets of $X$ as binary vectors. We note that linear algebra characterizations have been used to study the metric dimension of Hypercube graphs~\cite{Beardon:2013} and Hamming graphs~\cite{LaiTilBecLla20}. 

In what follows, we represent elements in $2^\curlyX$ as binary vectors of dimension $|X|$. Namely, for $a\in 2^\curlyX$, $a(x)=1$ when $x\in a$, and $a(x)=0$ when $x\notin a$. (For instance, $\curlyX$ is represented by a vector of all ones, whereas $\emptyset$ by a vector of all zeros.) Additionally, for $r\in2^\curlyX$ and $z\in\RR^{|\curlyX|}$, $\langle r,z\rangle$ denotes the inner product between the binary vector associated with $r$ and the vector $z$. Namely:
\[\langle r,z\rangle:=\sum_{x\in\curlyX}r(x)\cdot z(x).\]

In what follows, we use product notation to denote set intersections. Namely, if $a,b\in 2^X$ then $ab:=(a\cap b)$.

The next result characterizes equidistant sets in terms of inner products. This characterization will be used in Section~\ref{subsubsec:Sizing1}, in the proof of Theorem~\ref{thm:main1}, to assess the probability that two different subsets of $X$, of the same size, collide in terms of their distance to a random subset of $X$.

\begin{lemma}
Let $a,b,r\in 2^\curlyX$ and define the vector $z:=\big(|r|+|b|\big)\,a-\big(|r|+|a|\big)\,b$. If $\Jac(a,r)=\Jac(b,r)$ then $\langle r,z\rangle=0$. Conversely, if $r\ne\emptyset$ and $\langle r,z\rangle=0$ then $\Jac(a,r)=\Jac(b,r)$.
\label{lem:innerp}
\end{lemma}
\begin{proof}
We show first that
\begin{equation}
\langle r,z\rangle=\big(|r|+|b|\big)\cdot|ar|-\big(|r|+|a|\big)\cdot|br|.
\label{ide:innerzr}
\end{equation}
For this, observe that $|ar|=\langle a,r\rangle$ and $|br|=\langle b,r\rangle$; from which the identity in equation~(\ref{ide:innerzr}) is immediate due to the bilinearity of inner products.

Since $\langle \emptyset,z\rangle=0$, to complete the proof, it suffices to show that if $r\ne\emptyset$ then $\Jac(a,r)=\Jac(b,r)$ if and only if $\langle r,z\rangle=0$. For this, note that $1-\Jac(c,r)=|cr|/|c\cup r|$ and $|c\cup r|=|c|+|r|-\langle c,r\rangle$, for all $c\in 2^X$. In particular, a simple algebra  shows that $\Jac(a,r)=\Jac(b,r)$ is equivalent to having $(|r|+|b|)\,\langle a,r\rangle-(|r|+|a|)\,\langle b,r\rangle=0$, that is, $\langle z,r\rangle=0$ due to the bilinearity of inner products.
\end{proof}

We also want an inner product characterization of sets $a$ and $b$ that not only collide in terms of their Jaccard distance to a set $r$ but also to $r^c$, the complement of $r$. Our next result provides a necessary condition for both collisions to occur. This is characterization is used in Section~\ref{subsubsec:Sizing2} to show Theorem~\ref{thm:main2}.

\begin{corollary}
Let $a,b,r\in 2^\curlyX$. If $\Jac(a,r)=\Jac(b,r)$ and $\Jac(a,r^c)=\Jac(b,r^c)$ then $\big(|r^c|-|r|\big)\cdot\big(|br|-|ar|\big)=|r^c|\cdot\big(|b|-|a|\big)$.
\label{cor:innerp2}
\end{corollary}

\begin{proof}
If $\Jac(a,r)=\Jac(b,r)$ and $\Jac(a,r^c)=\Jac(b,r^c)$ then Lemma~\ref{lem:innerp} implies that $\langle r,z_1\rangle=0$ and $\langle r^c,z_2\rangle=0$, where $z_1:=\big(|r|+|b|\big)\,a-\big(|r|+|a|\big)\,b$ and $z_2:=\big(|r^c|+|b|\big)\,a-\big(|r^c|+|a|\big)\,b$. Hence, due to the identity in equation~(\ref{ide:innerzr}), we have that
\begin{align}
\nonumber 0
&=\langle r,z_1\rangle+\langle r^c,z_2\rangle\\
\nonumber &=\big(|r|+|b|\big)\cdot|ar|-\big(|r|+|a|\big)\cdot|br|+\big(|r^c|+|b|\big)\cdot|r^ca|-\big(|r^c|+|a|\big)|r^cb|\\
\label{ide:0equalsnice} &=|r|\cdot\big(|ar|-|br|\big)+|r^c|\cdot\big\{|r^ca|-|r^cb|\big\}+|ar|\cdot|b|-|br|\cdot|a|+\big\{|r^ca|\cdot|b|-|r^cb|\cdot|a|\big\}.
\end{align}
But $|r^ca|=|a|-|ar|$ and $|r^cb|=|b|-|br|$; in particular, we may rewrite the expressions within the curly parentheses above as follows: $|r^ca|-|r^cb|=\big(|br|-|ar|\big)+|a|-|b|$, and $|r^ca|\cdot|b|-|r^cb|\cdot|a|=|a|\cdot|br|-|b|\cdot|ar|$. Finally, substituting these two expressions back in equation~(\ref{ide:0equalsnice}), and after recognizing various terms cancellations, we obtain that
\[0=\big(|r^c|-|r|\big)\cdot\big(|br|-|ar|\big)-|r^c|\cdot\big(|b|-|a|\big),\]
from which the Corollary follows.
\end{proof}

\begin{table}[]
\centering
\begin{tabular}{|ccccccccccccccc|}
\hline
$|X|$ & 1  & 2  & 3 & 4 & 5  & 6 & 7  & 8 & 9  & 10 & 11  & 12 & 13  & 14 \\
\hline
$|R|$ & 1  & 2  & 2  & 3 & 3  & 4 & 5  & 5 & 6  & 6 & 7  & 7 & 8  & 8 \\
\hline
$\frac{1}{|R|}\sum\limits_{r\in R}|r|$ & 1  & 1.5  & 2.0  & 2.33 & 2.66  & 3.5 & 4.4  & 3.87 & 4.3  & 5.8 & 5.9  & 6 & 6.4  & 7.4 \\
\hline
\end{tabular}
\caption{Upper bounds for $\beta(2^X,\Jac)$ obtained by the ICH for a range of sizes of $X$. The middle row is the size of the resolving set $R$ found by the ICH. The bottom row is the average size of the sets in $R$, which often seem approximately equal to $|X|/2$.}
\label{tab:AvgResSize}
\end{table}

%%%%%%%%%%%%%%%%%%%%%%%%%%%%%%%%%%%%%%%
\subsection{Resolving Pairs of Subsets of $X$ of Equal Size}
\label{subsec:equalcard}
%%%%%%%%%%%%%%%%%%%%%%%%%%%%%%%%%%%%%%%

In this section, we prove Theorem~\ref{thm:main1} using the probabilistic method~\cite{AloSpe04}.

In what follows, $k\ge1$ and $r_1,\ldots,r_k$ are i.i.d. $\text{Bernoulli}(X,1/2)$ random subsets. In particular, $\EE|r_i|=|X|/2$, which is consistent with the experimental results displayed in Table~\ref{tab:AvgResSize}, and guided the selection of the parameter 1/2 in the Binomial distribution.

Define 
\[R_1:=\{r_1,\ldots,r_k\}.\]
In accordance with the probabilistic method, and to obtain an upper bound on the metric dimension of $2^X$, we aim to find a $k$ such that the probability that $R_1$ does not resolve all distinct pairs $a,b\in 2^X$ of equal size is strictly less than one. If we can find such a $k$, then there exists an $R\subset 2^X$ with $|R|=k$ that resolves all different pairs of subsets of $X$ of equal size. In particular, due to Lemma~\ref{lem:bummeryay}, we could assert that $\beta(2^X,\Jac)\leq(3+k)$. The challenge is to find $k$ as small as possible so that $(k+3)$ is a tight upper bound for the metric dimension of $2^X$, and the following probability 
\begin{equation}
\Sigma_1
:=\PP\left(\exists\,a,b\in 2^X,\text{ with } |a|=|b|\text{ but }a\ne b,\text{ such that }\forall r\in R_1:\Jac(a,r)=\Jac(b,r)\right)
\label{def:Sigma1}
\end{equation}
becomes asymptotically negligible as $|X|\to\infty$. Theorem~\ref{thm:main1} identifies a $k$ meeting this criterion. 

%%%%%%%%%%%%%%%%%%%%%%%%%%%%%%%%%%%%%%%
\subsubsection{Sizing the probability $\Sigma_1$}
\label{subsubsec:Sizing1}
%%%%%%%%%%%%%%%%%%%%%%%%%%%%%%%%%%%%%%%

In this section, we identify a $k$ in terms of $|X|$, of the same order of magnitude as the asymptotic lower bound for $\beta(2^X,\Jac)$ in Proposition~\ref{prop:lb1}, such that $\Sigma_1=o(1)$.

Suppose there exists $a,b\in 2^X$ such that $|a|=|b|$, $a\ne b$, and $\Jac(a,r)=\Jac(b,r)$ for all $r\in R_1$. Then, per Lemma~\ref{lem:innerp}, for each $r\in R_1$, $\big(|r|+|a|\big)\cdot\langle a-b,r\rangle=0$. But $|a|>0$ because $a\ne b$. So $\langle a-b,r\rangle=0$, or equivalently:
\[\langle z,r\rangle=0,\text{ with }z:=(ab^c-a^cb).\]
But observe that $z\in\mathcal{Z}$, where
\[\mathcal{Z}:=\left\{z\in\big\{0,\pm1\big\}^{|X|}\text{ such that }\sum_{x\in X}z(x)=0\text{ and }\sum_{x\in X}|z(x)|\ge2\right\}.\]
Consequently:
\[\Sigma_1
\le\PP\big(\exists\,z\in\mathcal{Z}\text{ such that }\forall r\in R_1:\langle z,r\rangle=0\big).\]

To bound the probability on the right-hand side above, consider a $z\in\mathcal{Z}$ and the sets $I:=\{x\in X\text{ such that }z(x)=+1\}$, and $J:=\{x\in X\text{ such that }z(x)=-1\}$. Observe that $I$ and $J$ are non-empty, disjoint, and of the same size; let $i$ be said cardinality. Note that $1\le i\le\lfloor|X|/2\rfloor$, and that $|I\cap r_1|,\ldots,|I\cap r_k|,|J\cap r_1|,\ldots,|J\cap r_k|$ are i.i.d. $\text{Binomial}(i,1/2)$ random variables. Further, since $\langle z,r_t\rangle=|I\cap r_t|-|J\cap r_t|$, $\langle z,r_1\rangle,\ldots,\langle z,r_k\rangle$ are also i.i.d. As a result:
\[\PP\big(\langle z,r_t\rangle=0\big)
=\sum_{j=1}^i{i\choose j}^2\left(\frac{1}{2}\right)^{2i}
=\left(\frac{1}{2}\right)^{2i}\sum_{j=1}^i{i\choose j}{i\choose i-j}
=\left(\frac{1}{2}\right)^{2i}\left\{{2i\choose i}-1\right\},\]
and
\begin{equation}
\Sigma_1
\le\sum_{i=1}^{\lfloor|X|/2\rfloor}{|X|\choose i,i,|X|-2i}\left\{{2i\choose i}\left(\frac{1}{2}\right)^{2i}\right\}^k.
\label{ine:Sigma21}    
\end{equation}

But 
\[{|X|\choose i,i,|X|-2i}\le\frac{|X|^{2i}}{(i!)^2}=O\left(\frac{\big(|X|e/i\big)^{2i}}{i}\right),\]
where the big-O is direct from Stirling's formula. On the other hand, Stirling's formula also implies that ${2i\choose i}\big(\frac{1}{2}\big)^{2i}\sim\frac{1}{\sqrt{i\pi}}$. However, for a bona fide substitution of ${2i\choose i}\big(\frac{1}{2}\big)^{2i}$ by $\frac{1}{\sqrt{i\pi}}$ in equation~(\ref{ine:Sigma21}), one needs a stronger relationship between these two sequences. For this effect, observe that~\cite{Rob55}:
\[\exp\left\{\frac{1}{12i+1}\right\}<\frac{i!}{\sqrt{2\pi}\,i^{i+1/2}\,e^{-i}}<\exp\left\{\frac{1}{12i}\right\},\text{ for all }i\ge1.\]
In particular,
\[{2i\choose i}\Big(\frac{1}{2}\Big)^{2i}
\le\frac{\exp\Big\{\frac{1-36i}{24i(12i+1)}\Big\}}{\sqrt{i\pi}}
\le\frac{1}{\sqrt{i\pi}},\]
and from the inequality in equation (\ref{ine:Sigma21}) we see that
\begin{equation}
\Sigma_1
=\frac{|X|\big(|X|-1\big)}{2^k}+O\left(\frac{1}{\pi^{k/2}}\sum_{i=2}^{\lfloor|X|/2\rfloor}\frac{\big(|X|e/i\big)^{2i}}{i^{k/2}}\cdot\frac{1}{i}\right).
\label{ine:Sigma22}
\end{equation}

The following result will let us handle the big-O term above.

\begin{lemma}
If $k\ge\frac{2|X|\ln(2e)}{\ln(|X|/2)}$ then $\frac{(|X|e/i)^{2i}}{i^{k/2}}=O(1)$, uniformly for all $|X|$ large enough and $2\le i\le\lfloor|X|/2\rfloor$.
\label{lem:Sigma2BigO1}
\end{lemma}

\begin{proof}
It suffices to show that
\begin{equation}
k\ge\frac{4i\ln\big(|X|e/i\big)}{\ln(i)},
\label{ine:kSigma2}
\end{equation}
for all $|X|$ large enough and $2\le i\le\lfloor|X|/2\rfloor$. For this, consider the function defined as $f(\tau):=4\tau\ln\big(|X|e/\tau\big)/\ln(\tau)$, for $2\le \tau\le|X|/2$. But note that
\begin{equation}
f'(\tau)
=4\frac{\ln(|X|e)\ln(\tau)-\ln^2(\tau)-\ln(|X|e)}{\big\{\ln(\tau)\big\}^2}
=4\frac{\ln\left(\frac{\tau_0}{\tau}\right)\cdot\ln\left(\frac{\tau}{\tau_1}\right)}{\big\{\ln(\tau)\big\}^2},
\label{ide:f'(t)4Sigma2}
\end{equation}
where the second identity assumes that $|X|>e^3$, in which case
\begin{align*}
\tau_0
&:=\exp\left\{\frac{1+\ln|X|}{2}\left(1-\sqrt{1-\frac{4}{1+\ln|X|}}\right)\right\}
=\exp\left\{1+O\left(\frac{1}{\ln|X|}\right)\right\};\\
\tau_1
&:=\exp\left\{\frac{1+\ln|X|}{2}\left(1+\sqrt{1-\frac{4}{1+\ln|X|}}\right)\right\}
=\exp\left\{\ln|X|+O\left(\frac{1}{\ln|X|}\right)\right\}.
\end{align*}
In particular, $\tau_0\sim e$ and $\tau_1\sim|X|$. Thus, as long as $|X|$ is large enough, $2<\tau_0<|X|/2<\tau_1$, and equation (\ref{ide:f'(t)4Sigma2}) implies that $f(\tau)$ is decreasing for $\tau\in[2,\tau_0]$ and increasing for $\tau\in[\tau_0,|X|/2]$, which in turn implies that $f(\tau)$ is maximized at $\tau=2$ or $\tau=|X|/2$. Since $f(2)\ll f\big(|X|/2\big)$, $f$ is maximized at $\tau=|X|/2$; in particular, the inequality in equation (\ref{ine:kSigma2}) is satisfied when $k\ge f\big(|X|/2\big)$, which shows the Lemma.
\end{proof}

Finally, due to equation (\ref{ine:Sigma22}), if $k$ satisfies the condition in Lemma~\ref{lem:Sigma2BigO1} then
\[\Sigma_1
=O\left(\frac{|X|^2}{2^k}\right)+O\left(\frac{1}{\pi^{k/2}}\sum_{i=2}^{\lfloor|X|/2\rfloor}\frac{1}{i}\right)
=o(1)+O\left(\frac{\ln|X|}{\pi^{k/2}}\right)=o(1),\]
where, for the middle identity, we have used that the harmonic series grows logarithmic with the number of terms. This completes the proof of Theorem~\ref{thm:main1}.

%%%%%%%%%%%%%%%%%%%%%%%%%%%%%%%%%%%%%%%
\subsection{Resolving Subsets of $X$ of Different Size}
\label{subsec:diffcard}
%%%%%%%%%%%%%%%%%%%%%%%%%%%%%%%%%%%%%%%

In this section, we prove Theorem~\ref{thm:main2}.

After having characterized the asymptotic order of the metric dimension of $(2^X,\Jac)$, i.e., the asymptotically optimal size of resolving sets for $2^X$, in this final section, we see how to resolve all pairs of subsets of $X$ of different size.

For this, consider the problem of resolving all distinct $a,b\in 2^X$ such that $|a|<|b|$, using a set of the form
\begin{equation}
\label{def:R2}
R_2=\{r_1,r_1^c,\ldots,r_k,r_k^c\},
\end{equation}
where $r_1,\ldots,r_k$ are i.i.d. with a $\text{Binomial}(X,1/2)$ distribution. It follows that
\[\PP\big(
R_2\text{ does not resolve all distinct }a,b\in2^X\text{ such that }|a|,|b|\le|X|/2\big)\le\Sigma_2,\]
where
\begin{equation}
\Sigma_2
:=\PP\left(\exists\,a,b\in 2^X,\text{ with } |a|<|b|\le\frac{|X|}{2},\text{ such that }\forall r\in R_2:\Jac(a,r)=\Jac(b,r)\right).
\label{def:Sigma2}
\end{equation}

%%%%%%%%%%%%%%%%%%%%%%%%%%%%%%%%%%%%%%%
\subsubsection{Sizing the probability $\Sigma_2$}
\label{subsubsec:Sizing2}
%%%%%%%%%%%%%%%%%%%%%%%%%%%%%%%%%%%%%%%

Suppose that $a,b,r\in 2^X$ are such that $|a|<|b|$, $\Jac(a,r)=\Jac(b,r)$, and $\Jac(a,r^c)=\Jac(b,r^c)$; in particular, per Corollary~\ref{cor:innerp2}, $|r^c|\cdot\big(|b|-|a|\big)=\big(|r^c|-|r|\big)\cdot\big(|br|-|ar|\big)$. But $|r^c|=|X|-|r|$, $|b|-|a|=|a^cb|-|b^ca|$, and $|br|-|ar|=|a^cbr|-|b^car|$. So, we may rewrite the last identity equivalently as follows:
\begin{equation*}
\big(|a^cb|-|b^ca|\big)\cdot\left(1-\frac{|r|}{|\curlyX|}\right)=\big(|a^cbr|-|b^car|\big)\cdot\left(1-\frac{2|r|}{|\curlyX|}\right).   
\label{ide:Cor2based}
\end{equation*}
Equivalently, if we define $\Delta_c(r):=|c|-2|cr|$ for each $c\in 2^X$, the above identity is equivalent to
\begin{equation}
\frac{\Delta_X(r)}{|X|}\cdot\frac{\Delta_{v}(r)-\Delta_{u}(r)}{|u|-|v|}=1,
\label{ide:mainintermezzo2}
\end{equation}
where $u:=a^cb$ and $v:=b^ca$ are disjoint subsets of $X$ such that $|u|>|v|$. Notably, for $r\sim\text{Binomial}(X,1/2)$, the probability of the above event depends only on the quantities $|u|$, $|v|$, and $|X|$, without regard to the specific identity of $u$ and $v$, except for the constraints that $uv=\emptyset$ and $|u|>|v|$. So we may define $\rho(i,j,X)$ as the probability of the event in (\ref{ide:mainintermezzo2})---when $(u,v)\in 2^X\times 2^X$ are such that $uv=\emptyset$, $|u|=i>j=|v|$, and $r\sim\text{Binomial}(X,1/2)$. 

It follows from the above discussion that if
\[\mathcal{P}:=\left\{(u,v)\in 2^X\times 2^X\text{ such that }uv=\emptyset\text{ and }|v|<|u|\right\},\]
then
\[\Sigma_2
\le
\PP\left(\exists\,(u,v)\in\mathcal{P}\text{ such that }\forall t\in\{1,\ldots,k\}:\frac{\Delta_X(r_t)}{|X|}\cdot\frac{\Delta_{v}(r_t)-\Delta_{u}(r_t)}{|u|-|v|}=1\right).\]
But note that the random vectors $\big(\Delta_u(r_t),\Delta_v(r_t),\Delta_X(r_t)\big)$, with $t=1,\ldots,k$, are i.i.d. for any given $(u,v)\in\mathcal{P}$. As a result:
\begin{equation}
\Sigma_2
\le\sum_{i=1}^{|X|}\sum_j{|X|\choose i,j,|X|-i-j}\,\rho^k\big(i,j,X\big),
\label{ine:almostthere}
\end{equation}
where the index $j$ is the inner sum above is such that $0\le j<i$ and $(i+j)\le|X|$.

\begin{lemma}
If $1\le i\le|X|$ and $0\le j<i$, with $(i+j)\le|X|$, then 
$\rho(i,j,X)\le4\exp\left(-\frac{\sqrt{|X|}}{2}\right)$.
\label{lem:ineqrho}
\end{lemma}

\begin{proof}
Let $(u,v)\in\mathcal{P}$ be such that $|u|=i$ and $j=|v|$; in particular, $i>j$. Then, for each $\tau>0$:
\begin{align*}
\rho(i,j,X)
&=\PP\left(\Delta_X(r)\cdot\big(\Delta_{v}(r)-\Delta_{u}(r)\big)=(i-j)\cdot|X|\right)\\
&\le\PP\left(|\Delta_X(r)|\ge\sqrt{2\tau|X|}\right)+\PP\left(|\Delta_{u}(r)-\Delta_{v}(r)|\ge(i-j)\sqrt{\frac{|X|}{2\tau}}\right)\\
&\le 2\left\{\exp\left(-\tau\right)+\exp\left(-\frac{(i-j)^2|X|}{4(i+j)\tau}\right)\right\},
\end{align*}
where for the last inequality we have used the well-known Hoeffding's inequality, and that $2\big(\Delta_{u}(r)-\Delta_{v}(r)\big)$ has the same distribution as $\sum_{k=1}^{i+j}\big(Z_k-\EE(Z_k)\big)$, where $Z_1,\ldots,Z_{i+j}$ are independent random variables, with $Z_k\sim\text{Bernoulli}(1/2)$ for $1\le k\le i$, and $(-Z_k)\sim\text{Bernoulli}(1/2)$ for $i<k\le i+j$. Therefore, by selecting
\begin{equation}
\tau:=\frac{i-j}{2}\sqrt{\frac{|X|}{i+j}}
\label{def:opttau1}
\end{equation}
we obtain that 
\begin{equation}
\rho(i,j,X)\le 4e^{-\tau}.
\label{def:opttau2}
\end{equation}
But
\begin{equation}
\tau=\sqrt{i}\cdot\frac{1-j/i}{\sqrt{1+j/i}}\cdot\frac{\sqrt{|X|}}{2} \ge\frac{\sqrt{|X|}}{2},
\label{def:opttau3}
\end{equation}
because the first factor above is an increasing function of $i$, whereas the second factor is a decreasing function of $j/i$. The lemma is now a direct consequence of the inequalities in (\ref{def:opttau2})-(\ref{def:opttau3}).
\end{proof}

\begin{rmk}
The choice of $\tau$ in (\ref{def:opttau1}) is somewhat optimal when $\tau\ge1$, which is a necessary condition for the upper-bound in (\ref{def:opttau2}) to be non-trivial. (The latter requires of course $\tau\ge 2\ln 2$ which, based on (\ref{def:opttau3}), can be guaranteed as soon as $|X|\ge8$.) Indeed, from the last proof: $\rho(i,j,X)\le2 f(t)$, where $f(t):=e^{-t}+e^{-\tau^2/t}$ for $t>0$. But note that $f'(t) 
=\left(g(\tau^2/t)-g(t)\right)/t$, with $g(t):=t e^{-t}$ for $t>0$; hence $t=\tau$ is a critical point of $f(t)$. Moreover, since
$f''\big(t\big) =2\,e^{-t}\left(1-t^{-1}\right)$, $t=\tau$ is a local minimum when $\tau>1$. In particular, since $f'''(1)=0$ but $f''''(1)=2e^{-1}>0$ when $\tau=1$, $t=\tau$ is a local minimum of $f(t)$ when $\tau\ge1$. 
\end{rmk}

Let $\rho_X$ be the upper bound for $\rho(i,j,X)$ given in Lemma~\ref{lem:ineqrho}. It follows from (\ref{ine:almostthere}) that
\begin{align*}
\Sigma_2
&\le\rho^k_X\sum_{i=1}^{|X|}\sum_j{|X|\choose i,j,|X|-i-j}\\
&\le\rho^k_X\sum_{i=1}^{|X|}\sum_j\frac{|X|^{i+j}}{i!\,j!}\\
&\le\rho^k_X\left(\sum_{i=1}^{|X|}\frac{|X|^i}{i!}\right)^2\\
&=\rho^k_X e^{2|X|}\,\left(\frac{\Gamma\left(|X|+1,|X|\right)}{|X|!}\right)^2,
\end{align*}
where, for an integer $n>0$ and $x\in\RR$, $\Gamma(n,x):=(n-1)!\,e^{-x}\sum\limits_{i=0}^{n-1}\frac{x^k}{(n-1)!}=\int\limits_x^\infty t^{n-1}e^{-t}\,dt$ is the (upper) incomplete Gamma function. Finally, due to~\cite[Proposition 2.7]{Pin23}, $\Gamma(|X|+1,|X|)=O\big(|X|^{|X|}e^{-|X|}\big)$. Consequently,
\[\Sigma_2=O\left(\frac{\rho_X^k \,|X|^{2|X|}}{(|X|!)^2}\right)=O\left(\frac{\rho_X^k e^{2|X|}}{|X|}\right)=O\left(\frac{e^{2|X|+k\ln(\rho_X)}}{|X|}\right),\]
where we have used the Stirling's approximation and the exp-log transform. In particular, for any $\epsilon<1$, if select $k$ so that $2|X|+k\ln(\rho_X)\le\epsilon\ln|X|$, for instance, $k=\left\lceil\frac{4|X|-2\epsilon\ln|X|}{\sqrt{|X|}-4\ln2}\right\rceil\sim4\sqrt{|X|}$, then $\Sigma_2=o(1)$, which completes the proof of Theorem~\ref{thm:main2}.

%%%%%%%%%%%%%%%%%%%%%%%%%%%%%%%%%%%%%%%
\subsection{Resolving Comparatively Small Subsets of $X$}
\label{subsec:smallcard}
%%%%%%%%%%%%%%%%%%%%%%%%%%%%%%%%%%%%%%%

In this section we prove Corollary~\ref{cor:main3}, which is the consequence of arguments already used in the proofs of theorems~\ref{thm:main1} and~\ref{thm:main2}. For this, let $0<\epsilon<1$, and $1\le W\le (1-\epsilon)(\ln\pi)\sqrt{|X|}/\ln|X|$ be an integer. 

To show the Corollary, we reconsider the set $R_2$ in (\ref{def:R2}) with $k\ge(4+\epsilon)\sqrt{|X|}$. By distinguishing pairs $a,b\in 2^X$ such that $|a|=|b|$ from $|a|\ne|b|$, we find  this time that
\begin{equation}
\PP\left(\exists\,a,b\in 2^X\text{ with }a\ne b\text{ such that }\forall r\in R_2:\Jac(a,r)=\Jac(b,r)\right)\le\Sigma_2+\Sigma_3,
\label{ide:last}
\end{equation}
where $\Sigma_2$ is the double-sum in (\ref{ine:almostthere}), and  $\Sigma_3$ is a truncated version of the summation in (\ref{ine:Sigma21}). Specifically
\[\Sigma_3:=\sum_{i=1}^{W}{|X|\choose i,i,|X|-2i}\left\{{2i\choose i}\left(\frac{1}{2}\right)^{2i}\right\}^k.\]
But, from the discussion in Section~\ref{subsubsec:Sizing2}, we already know that $\Sigma_2=o(1)$. On the other hand, from the discussion in Section~\ref{subsubsec:Sizing1} that led to (\ref{ine:Sigma22}), we can say that
\[\Sigma_3= O\left(\sum_{i=1}^{W}\frac{|X|^{2i}}{(i!)^2\,(i\pi)^{k/2}}\right).\]
As a result
\[\Sigma_3
= O\left(\frac{|X|^{2W}}{\pi^{k/2}}\sum_{i=1}^{W}\frac{1}{(i!)^2}\right) 
= O\left(\frac{|X|^{2W}}{\pi^{k/2}}\right) 
= O\left(\frac{|X|^{2W}}{\pi^{2\sqrt{|X|}}}\right)
= O\left(\pi^{-2\epsilon\sqrt{|X|}}\right),\]
where for the last two asymptotic bounds we have use the constrains on $k$ and $W$. The Corollary is now a direct consequence of the inequality in (\ref{ide:last}).

\newpage 

\hfill\\
\noindent\textbf{Acknowledgments.} This work was partially funded by the NSF grant No. 1836914.

%\bibliography{3references.bib}

\end{document}